\documentclass[10pt]{article}
\usepackage{}

\usepackage{pgf}
\usepackage[square, comma, sort&compress, numbers]{natbib}
\usepackage[english]{babel}
\usepackage[utf8]{inputenc}
\usepackage{booktabs}
\usepackage{caption}
\usepackage[T1]{fontenc}
\usepackage[font=small,labelfont=bf,tableposition=top]{caption}
\usepackage{threeparttable}

\usepackage{datetime}
\usepackage{geometry}
\usepackage{array}
\usepackage{xspace}

\usepackage{amsmath,amsfonts,amssymb,amsopn,amsthm,amscd}
\theoremstyle{plain}
\usepackage{graphicx} 
\usepackage{epstopdf}
\usepackage{enumerate}

\def\Box{\vcenter{\vbox{\hrule\hbox{\vrule
     \vbox to 8.8pt{\hbox to 10pt{}\vfill}\vrule}\hrule}}}

\usepackage{indentfirst}
\newcommand{\Ff}{{\mathbb F}}

\newtheorem{thm}{Theorem}[section]
\newtheorem{lem}[thm]{Lemma}

\newtheorem{remark}{Remark}

\date{}

\begin{document}

%

\title{Multisequences with high joint nonlinear complexity from function fields \thanks{The work was supported by National Science Foundation of
China No. 61602342,  Natural Science Foundation of Tianjin  under grant No. 18JCQNJC70300,  NFSC 11701553,   the Science and
Technology Development Fund of Tianjin Education Commission
for Higher Education No. 2018KJ215, KYQD1817,  and the China Scholarship Council (No. 201809345010 and No. 201907760008), Key Laboratory of Applied Mathematics of Fujian Province University
(Putian University) (No. SX201904 and SX201804), NSFT  No.16JCYBJC42300,  NFSC No. 61872359, 61972456 and 61802281, the Science and
Technology Development Fund of Tianjin Education Commission
for Higher Education No.  2017KJ213.}
}


\author{Yang Yan\thanks{Y. Yan is with the School of Information Technology and Engineering, Tianjin University of Technology and Education, Tianjin, 300387, Email: yanyangucas@126.com},
Qiuyan Wang\thanks{Q. Wang is with the School of Computer Science and
Technology, Tianjin Polytechnic University, Tianjin,
300387, China, and with the Provincial Key Laboratory of Applied Mathematics, Putian University, Putian, Fujian 351100, China. Email: wangyan198801@163.com},
Chenhuang Wu\thanks{C. Wu is with the Provincial Key Laboratory of Applied Mathematics, Putian University, ptuwch@163.com}
}
\maketitle

\let\thefootnote\relax\footnotetext{}

\begin{abstract}
Multisequences over finite fields play a pushing role in the applications that relate to parallelization, such as word-based stream ciphers and pseudorandom vector generation. It is interesting to study the complexity measures for multisequences. In this paper, we propose three constructions of multisequences over finite fields from rational function fields and Hermitian function fields. We also analyze the joint nonlinear complexity of these multisequences. Moreover, the length and dimension of these multisequences are flexible.

{\bf Keywods}: Sequence, multisequence, joint nonlinear complexity, function field
\end{abstract}

\section{Introduction}
The study of pseudorandom sequences is a hot research topic, due to their utilization in the generation of pseudorandom numbers and cryptography. The performance of a pseudorandom sequence is determined by complexity-theoretic and statistical requirements. In practical applications, as complexity-theoretic and statistical requirements are in a sense independent \cite{Nie1}, these two requirements are both important.

To assess the capability of a keystream generated by a stream cipher, one has to consider that replicating the entire keystream from a part of the keystream should be very hard. To this end, it is interesting to know how hard a pseudorandom sequence might be to replicate, which leads that many scholars investigate pseudorandom sequences from the complexity-theoretic standpoint. Several complexity measures for sequences are available in the reference therein. The most popular complexity measure is the linear complexity where only linear feedback shift registers are considered. A concise survey on the linear complexity has been provided in \cite{Win1} and the recent handbook article \cite{Meidl1}. However, a few effort has devoted to the complexity measure referring to feedback shift registers with feedback functions of higher algebraic degree, which is called the nonlinear complexity (see \cite{Luo, Nie2}). As a special type of the nonlinear complexity, the maximum-order complexity has attracted some attention due to Jansen \cite{J1,J2}. Basing on pattern counting, complexity measures for sequences were established, for instance, the Lempel-Ziv complexity (see \cite{Lem1} for the definition and \cite{Lem2} for cryptographic applications).

For applications that relate to parallelization, such as word-based stream ciphers and pseudorandom vector generation, multisequences over finite fields are indispensable (see \cite{M1,M2}). The complexity study of multisequences has focused on the joint linear complexity and $k$-error linear complexity \cite{Mul1,Mul2,Mul3,Mul4,Mul5,Mul6,Mul7,Meidl1,Xing1}. Recently, Meidl and Niederreiter \cite{Nonlinear} introduced the definition of the joint nonlinear complexity for multisequences (see Section \ref{S2}). In practice, a multisequence may have large joint linear complexity, but very small joint nonlinear complexity. Hence, we would like to construct multisequences with high joint nonlinear complexity. In fact, the design of multisequences with high joint nonlinear complexity is harder than that of multisequences with high joint linear complexity.

Algebraic function fields (or algebraic curves) over finite fields are powerful tools to construct a variety of sequences. For example, sequences with low correlations were proposed in \cite{Hu1,Xing2}; sequences and multisequences with large linear complexity were present in \cite{Xing1,Xing3,Xing4}; the authors \cite{Nie2,Luo} has constructed sequences with high nonlinear complexity.

The purpose of this paper is to construct multisequences with high joint nonlinear complexities. Using rational function fields and Hermitian function fields which contain automorphisms with large order, we propose three constructions of multisequences with flexible lengths and dimensions. Additionally, we give the lower bound on the joint nonlinear complexities for these multisequences. Comparing with the behavior of joint nonlinear complexities of random multisequences, these multisequences can be said to have high joint nonlinear complexity under certain conditions on their parameters.

This paper is organized as follows. Section \ref{S2} devotes to some definitions and results about the joint nonlinear complexity and function fields. In Section \ref{S3} and Section \ref{S4}, we propose three construction of multisequences and evaluate the lower bound of the joint nonlinear complexity. Section \ref{S5} concludes the paper.

\section{Preliminaries}\label{S2}

In this section, we  briefly recall some basic definitions and results about multisequences and function fields, which will be needed in our discussion. We begin with the background on the nonlinear complexity of a multisequence.

\subsection{Multisequences and joint nonlinear complexity}
Throughout this paper, let $q$ be a power of an arbitrary prime $p$ and $\Ff_q$ stand for the finite field with $q$ elements. We write $\Ff_q^*=\Ff_q\setminus\{0\}$. For any positive integer $u$, denote by $\Ff_q[x_1,\cdots,x_u]$ the ring of polynomials of $\Ff_q$ with the $u$ variables $x_1,\cdots,x_u$.

Assume that $\mathbf{s}=\{s(j)\}_{j=0}^{N-1}$ is a nonzero sequence of length $N$ over $\Ff_q$. We say that a polynomial $f\in\Ff_q[x_1,\cdots,x_u]$ generates the sequence $\mathbf{s}$ if
$$
s(j+u)=f(s(j),s(j+1),\cdots,s(j+u-1)),
$$
for any $j=0,1,\cdots,N-u-1$.

Suppose that $r$ is a positive integer. The $r$th-order nonlinear complexity $N_r(\mathbf{s})$ of $\mathbf{s}$ is the smallest integer $u\geq1$ such that there exists a polynomial $f\in\Ff_q[x_1,\cdots,x_u]$ of degree at most $r$ in each variable that generates $\mathbf{s}$. Furthermore, if $\mathbf{s}$ is the zero sequence, then the nonlinear complexity $N_r(\mathbf{s})$ is equal to $0$.

For an integer $M\geq1$, let $\mathcal{S}=\{\mathbf{s}_i=\{s_i(j)\}_{j=0}^{N-1}:i=1,2,\cdots,M\}$ be a set of $M$ nonzero sequences of length $N$ over $\Ff_q$. Then $\mathcal{S}$ is called a multisequence of dimension $M$ over $\Ff_q$. The $r$th-order joint nonlinear complexity $N_r(\mathcal{S})$ of $\mathcal{S}$ is defined to be the smallest integer $u\geq1$ such that there exists a polynomial $f\in\Ff_q[x_1,\cdots,x_u]$ of degree at most $r$ in each variable that generates all $M$ sequences in $\mathcal{S}$ simultaneously. Moreover, $N_r(\mathcal{S})$ is set to be $0$ if $\mathbf{s}_i$ is the zero sequence for any $1\leq i \leq M$ and $N_r(\mathcal{S})$ is defined to be $N$ if there is no such polynomial generating the $N$ terms of each sequence in $\mathcal{S}$ simultaneously.

According to the definition of the joint nonlinear complexity, we always have $0\leq N_r(\mathcal{S})\leq N$. As point out in \cite{Nonlinear}, it suffices to consider the case that $1\leq r\leq q-1$ in the definition of $N_r(\mathcal{S})$. When $r\geq q-1$, all joint nonlinear complexities of a certain $\mathcal{S}$ are equal to $N_{q-1}(\mathcal{S})$. If $r=q-1$ and the set $\mathcal{S}$ contains only one sequence, i.e., $M=1$, the nonlinear complexity $N_{q-1}(\mathcal{S})$ is equal to the maximum-order complexity introduced by Jansen \cite{J1,J2}. For $M>1$ and $r=q-1$, we may term $N_{q-1}(\mathcal{S})$ the joint maximum-order complexity of $\mathcal{S}$. The definition of \cite[Definition 1]{J2} may regard as a previous notion of the joint maximum-order complexity.

For a set $\mathcal{S}=\{\mathbf{s}_i=\{s_i(j)\}_{j=0}^{\infty}:i=1,2,\cdots,M\}$ of infinite sequences over $\Ff_q$, we define the joint nonlinear complexity of $\mathcal{S}$ by $N_r(\mathcal{S},n)=N_r(\mathcal{S}_n)$, where $\mathcal{S}_n=\{\mathbf{s}_i=\{s_i(j)\}_{j=0}^{n-1}:i=1,2,\cdots,M\}$ and $n$ is a positive integer.

\subsection{Some background on function fields}

A function field $F$ over $\Ff_q$ is an extension field of $\Ff_q$ such that $F$ is a finite extension of $\Ff_q(x)$ for some element $x\in F$ which is transcendental over $\Ff_q$. In the following of this subsection, we always suppose that $\Ff_q$ is the full constant field of $F$.

For a discrete valuation $v$ which maps $F$ to $\mathbb{Z}\cup \{\infty\}$, define a local ring of $F$ by $\mathcal{O}=\{z\in F: v(z)\geq 0\}$ and its unique maximal ideal $P$ is termed a place of $F$. Denoted by $v_P$ and $\mathcal{O}_P$ the discrete valuation and the local ring associated with $P$, respectively. The residue class field $\mathcal{O}_P/P$ is a finite extension of $\Ff_q$ and the extension degree is called the degree of $P$, denoted by ${\rm deg}(P)$. Furthermore, a place $P$ is said to be a rational place if ${\rm deg}(P)=1$.

Assume that $\mathbb{P}_F$ is the set of all places of $F$. Let $S$ be a finite subset of $\mathbb{P}_F$. A divisor $D$ of $F$ is a formal sum
$$
D=\sum_{P\in S}m_PP,
$$
where $m_P$ is an integer for any $P\in S$. Define the degree of $D$ by
$$
{\rm deg}(D):=\sum_{P\in S}m_P{\rm deg}(P).
$$
Let $z$ be a nonzero function of $F$. Then the zero divisor and the pole divisor of $z$ are defined by
$$
(z)_0:=\sum_{P\in \mathbb{P}_F, v_P(z)>0 }v_P(z)P,
$$
and
$$
(z)_\infty:=-\sum_{P\in \mathbb{P}_F, v_P(z)<0 }v_P(z)P,
$$
respectively. Clearly, the principal divisor $(z)=(z)_0-(z)_\infty$. The degree of $(z)$ is $0$ due to the fact that ${\rm deg}((z)_0)={\rm deg}((z)_\infty)$ \cite[Threorem 1.4.11]{AG}.

For a divisor $D$ of $F$, the Riemann-Roch space is formed by
$$
\mathcal{L}(D)=\{z\in F\setminus \{0\} : (z)+D \geq 0\}\cup\{0\}.
$$
It is well known that $\mathcal{L}(D)$ is a finite dimensional space over $\Ff_q$. Let ${\rm dim}_{\Ff_q}\mathcal{L}(D)$ stand for the dimension of $\mathcal{L}(D)$. From the Riemann-Roch Theorem \cite{AG}, we obtain
$$
{\rm dim}_{\Ff_q}\mathcal{L}(D)\geq {\rm deg}(D)+1-g,
$$
where $g$ is the genus of $F$. In addition, the equality holds if ${\rm deg}(D)\geq 2g-1$.

Let $\varphi$ be an automorphism of $F$ which preserves all elements of $\Ff_q$, namely, $\varphi(a)=a$ for any $a\in\Ff_q$. All such automorphisms form a group of automorphisms of $F$ over $\Ff_q$ that is denoted by
${\rm Aut}(F/\Ff_q)$. The following lemma provides some basic properties on the automorphisms of $F$.
\begin{lem}\label{lem1}\cite{AG}
Assume that $P$ is a place of $F$ and $z$ is a function of $F$. For any $\varphi\in{\rm Aut}(F/\Ff_q)$, we have

1) $\varphi(P)$ is still a place of $F$ and ${\rm deg}(\varphi(P))={\rm deg}(P)$;

2) $v_{\varphi(P)}(\varphi(z))=v_P(z)$;

3) $\varphi(z)(\varphi(P))=z(P)$ if $v_P(z)\geq 0$.

\end{lem}

For more details on function fields, we refer the reader to the book \cite{AG}.

\subsection{Rational function fields}

Let $\Ff_q(x)$ be the rational function field over $\Ff_q$. As is known, $\Ff_q(x)$ has the genus $g=0$ and $q+1$ rational places. For each $\varphi\in{\rm Aut}(\Ff_q(x)/\Ff_q)$, there exist $a,b,c,d\in\Ff_q$ such that $\varphi(x)=\frac{ax+b}{cx+d}$ and $ad\neq bc$. It is easy to check that the order of ${\rm Aut}(\Ff_q(x)/\Ff_q)$ is $q^3-q$. Let $P_\infty$ be the unique pole of $x$. For every $\omega\in\Ff_q$, there is a unique rational place $P_{\omega}$ of $\Ff_q(x)$ with $x(P_{\omega})=\omega$. Then, we have the following lemma.
\begin{lem}\label{lem2}
Let $d>0$ be an integer with $d|(q-1)$ and $\alpha$ a nonzero element of $\Ff_q^*$ of order $d$. Assume that $\varphi$ is an automorphism of ${\rm Aut}(\Ff_q(x)/\Ff_q)$ such that $\varphi(x)=\alpha x$. Then,

1) $G=\{\varphi^i:i=0,1,\cdots,d-1\}$ is a cyclic group of order $d$

2) $\varphi(P_\infty)=P_\infty$ and $\varphi(P_0)=P_0$;

3) If $P$ is a rational place of $\Ff_q(x)$ such that $P\neq P_\infty, P_0$, then the cardinality of the set $\{\varphi^i(P):i=0,1,\cdots,d-1\}$ is $d$. Furthermore, the action of $G$ on all rational places of $\Ff_q(x)$ gives rise to $2+\frac{q-1}{d}$ among which two orbits contain one element and each of the other orbits contains $d$ element.
\end{lem}
The proof of this lemma is obvious and we omit it.

\subsection{Hermitian function fields}

The Hermitian function field over $\Ff_{q^2}$ is given by
$$
H=\Ff_{q^2}(x,y)\ \ \ \ {\rm and}\ \ \ y^q+y=x^{q+1},
$$
where $x,y$ are two variables over $\Ff_{q^2}$. The genus of $H$ is $g=\frac{q^2-q}{2}$. There are altogether $q^3+1$ rational places of $H$, namely the common pole $Q_{\infty}$ of $x$ and $y$ and $P_{a,b}$ with $x(P_{a,b})=a$ and $y(P_{a,b})=b$, where $a,b$ satisfy the equation $b^q+b=a^{q+1}$. The Hermitian function field is said to be a maximal function field since it meets the Hasse-Weil bound $1+q^2+2gq=q^3+1$.

The automorphism group $\mathcal{G}={\rm Aut}(H/\Ff_{q^2})$ has been completely determined in \cite{auto1,auto2} and it is isomorphic to the projective unitary group $PGU(3,q^2)$ of order $q^3(q^2-1)(q^3+1)$. Let $P_{\infty}$ be the unique pole of $x$. Then $\mathcal{G}_{\infty}$ in which all automorphisms fix $P_{\infty}$ is a subgroup of $\mathcal{G}$. Precisely speaking,
\begin{eqnarray*}
\mathcal{G}_{\infty}&=&\{\varphi\in\mathcal{G}:\varphi(P_{\infty})=P_{\infty}\}\\
&=&\{\varphi_{\alpha,\beta,\gamma}:\alpha\in\Ff_{q^2}^*,\beta,\gamma\in\Ff_{q^2},\gamma^q+\gamma=\beta\},
\end{eqnarray*}
where $\varphi_{\alpha,\beta,\gamma}$ stands for the automorphism
$$
\varphi_{\alpha,\beta,\gamma}(x)=\alpha x+\beta,\ \ \ \varphi_{\alpha,\beta,\gamma}(y)=\alpha^{q+1}x+\alpha \beta^qx+\gamma.
$$
It is easy to see that the order of $\mathcal{G}_{\infty}$ is $q^3(q^2-1)$ since $\alpha$ is arbitrary in $\Ff_{q^2}^*$ and $\gamma$ has $q$ solutions for any $\beta\in\Ff_{q^2}$.

Let $\delta$ be a primitive element of $\Ff_{q^2}$. Then, $\sigma=\varphi_{\delta,0,0}$ is the automorphism with order $q^2-1$ and it generates a cyclic group $R$ of order $q^2-1$. Obviously, $R$ is a subgroup of $\mathcal{G}_{\infty}$. The following result plays a pushing role in the design of multisequences with high joint nonlinear complexity.
\begin{lem}\cite{Xing1}\label{lem3}
Let $R$ be a cyclic group defined as above. Then the action of $R$ on all rational places $\neq P_\infty$ of $H$ gives rise to $q+2$ orbits among which one orbit contains only one element, one orbit contains $q-1$ elements and each of the rest orbits contains $q^2-1$ elements.
\end{lem}

Assume that $\theta\in\Ff_{q^2}$ is a nonzero element such that $\theta^q+\theta=0$. Then the automorphism $\phi=\varphi_{1,0,\theta}$ has the order $p$ and the cyclic group $G$ formed by $\phi$ is of order $p$. The last auxiliary result which will be used is the following.
\begin{lem}\cite{Xing1}\label{lem4}
Suppose that $G$ is a cyclic group defined as above. Then $G$ acts on all rational places $\neq P_\infty$ of $H$ giving rise to $q^3/p$ orbits and every orbit contains $p$ distinct elements.
\end{lem}

\section{Multisequences from rational function fields}\label{S3}

In this section, employing some results on rational function fields, we propose a construction of multisequences and determine the lower bound of the joint nonlinear complexity.

Let $F=\Ff_q(x)$ be the rational function field over $\Ff_q$. According to Lemma \ref{lem2}, under the action of $G$ on all the rational places of $F$, there are $\frac{q-1}{d}$ orbits among which each orbits contains $d$ distinct rational places. Let $d>1$ and $3 \leq \frac{q-1}{d}$. Label all the elements of these $\frac{q-1}{d}$ orbits
\begin{eqnarray*}
&&P,\varphi(P),\cdots,\varphi^{d-1}(P);\\
&&Q_{1,1},\varphi(Q_{1,1}),\cdots,\varphi^{d-1}(Q_{1,1});\\
&&Q_{1,2},\varphi(Q_{1,2}),\cdots,\varphi^{d-1}(Q_{1,2});\\
&&\ \ \ \ \ \ \ \ \ \ \ \ \ \ \ \ \ \vdots\\
&&Q_{1,M},\varphi(Q_{1,M}),\cdots,\varphi^{d-1}(Q_{1,M});\\
&&\ \ \ \ \ \ \ \ \ \ \ \ \ \ \ \ \ \vdots\\
&&Q_{N,M},\varphi(Q_{N,M}),\cdots,\varphi^{d-1}(Q_{N,M}),\\
\end{eqnarray*}
where $NM=\frac{q-1}{d}-1$ and $N\geq 1$, $M\geq 1$. Assume that $z$ is a function of $F$ such that $(z)_\infty=P$. For $1\leq i\leq N$, we define a sequence of length $dM$ as
$$
\mathbf{s}_i=\{s_i(j)\}_{j=0}^{dM-1}=(z(Q_{i,1}),z(\varphi(Q_{i,1}))\cdots,z(\varphi^{d-1} (Q_{i,1})),\cdots,z(Q_{i,M}),\cdots,z(\varphi^{d-1} (Q_{i,M}))),
$$
and then define a set by
\begin{equation}\label{con1}
\mathcal{S}=\{\mathbf{s}_i:i=1,2,\cdots,N\}.
\end{equation}
Thus, $\mathcal{S}$ is a multisequence over $\Ff_q$ of dimension $N$. Additionally, $\mathcal{S}$ is periodic with least period $d$ if $M=1$.

\begin{thm}\label{thm1}
Let $d$ be an integer and $q$ a prime power such that $d>1$, $d|(q-1)$ and $3 \leq \frac{q-1}{d}$. Assume that $N,M$ are positive integers with $NM=\frac{q-1}{d}-1$. Let $\mathcal{S}$ be the multisequence defined by (\ref{con1}). Then, for any integer $r$ with $1\leq r\leq q-1$, we have
$$
N_r(\mathcal{S}_n)\geq\left\{
            \begin{array}{ll}
              \frac{nN-1}{N+r},& \hbox{if\ $1<n\leq d$,}\\
              \frac{dN\lfloor\frac{n}{d}\rfloor-1}{N\lfloor\frac{n}{d}\rfloor+r},& \hbox{if\ $d\leq n\leq dM$.}
            \end{array}
          \right.
$$
\end{thm}
\begin{proof}
Our first goal is to show that $\mathbf{s}_i$ is not a zero sequence for $1\leq i\leq N$. By the definition of $\mathbf{s}_i$, we deduce that ${\rm deg}((z)_0)={\rm deg}((z)_\infty)=1$ since $(z)_\infty=P$ and ${\rm deg}((z)_\infty)=1$. If $\mathbf{s}_i$ is a zero sequence, then there exist at least $q-1-d$ rational places which are zeros of $z$. Due to $q-1-d>1$, we obtain ${\rm deg}((z)_0)>1$ which leads to a contradiction. Hence, for $1\leq i\leq N$, $\mathbf{s}_i$ is not a zero sequence and $N_r(\mathcal{S}_n)\geq 1$.

Let $n>1$ be an integer. Assume that $f\in\Ff_q[x_1,\cdots,x_u]$ with $1\leq u\leq n-1$ is a polynomial of degree at most $r$ in every variable satisfying
\begin{equation}\label{thme1}
s_i(j+u)=f(s_i(j),s_i(j+1),\cdots,s_i(j+u-1)),
\end{equation}
for $0\leq j\leq n-u-1$ and $1\leq i\leq N$. In order to determine the lower bound of $N_r(\mathcal{S}_n)$, we divide the computation into two cases.

$Case\ 1:$ If $n\leq d$, it follows from (\ref{thme1}) that
$$
z(\varphi^{j+u}(Q_{i,1}))-f(z(\varphi^{j}(Q_{i,1})),z(\varphi^{j+1}(Q_{i,1})),\cdots,z(\varphi^{j+u-1}(Q_{i,1})))=0,
$$
for $0\leq j\leq n-u-1$ and $1\leq i\leq N$. Using Lemma \ref{lem1}, we have
\begin{eqnarray*}
&&z(\varphi^{j+u}(Q_{i,1}))-f(z(\varphi^{j}(Q_{i,1})),z(\varphi^{j+1}(Q_{i,1})),\cdots,z(\varphi^{j+u-1}(Q_{i,1})))\\
&=&\varphi^{-u}(z)(\varphi^{j}(Q_{i,1}))-f(z(\varphi^{j}(Q_{i,1})),\varphi^{-1}(z)(\varphi^{j}(Q_{i,1})),\cdots,\varphi^{-u+1}(z)(\varphi^{j}(Q_{i,1})))\\
&=&\left(\varphi^{-u}(z)-f(z,\varphi^{-1}(z),\cdots,\varphi^{-u+1}(z))\right)(\varphi^{j}(Q_{i,1})),
\end{eqnarray*}
and so we get that
\begin{equation}\label{thme2}
\left(\varphi^{-u}(z)-f(z,\varphi^{-1}(z),\cdots,\varphi^{-u+1}(z))\right)(\varphi^{j}(Q_{i,1}))=0,
\end{equation}
for $0\leq j\leq n-u-1$ and $1\leq i\leq N$.

We write $g=\varphi^{-u}(z)-f(z,\varphi^{-1}(z),\cdots,\varphi^{-u+1}(z))$. Note that $(z)_\infty=P$. Then we obtain $v_P(z)=-1$. By Lemma \ref{lem1}, we deduce that $v_{\varphi^{-u}(P)}(\varphi^{-u}(z))=-1$. Clearly, for any $0\leq t\leq u-1$, $\varphi^{-u}(P)$ is not a pole of $\varphi^{-t}(z)$. Consequently, we have
$$
v_{\varphi^{-u}(P)}(f(z,\varphi^{-1}(z),\cdots,\varphi^{-u+1}(z)))\geq 0
$$
and
$$
v_{\varphi^{-u}(P)}(f(z,\varphi^{-1}(z),\cdots,\varphi^{-u+1}(z)))\neq v_{\varphi^{-u}(P)}(\varphi^{-u}(z)),
$$
which implies that $\varphi^{-u}(z)\neq f(z,\varphi^{-1}(z),\cdots,\varphi^{-u+1}(z))$, namely, $g\neq 0$. From the discussion above, it is easy to see that
$$
g\in\mathcal{L}\left(\varphi^{-u}(P)+r\sum_{i=0}^{u-1}\varphi^{-i}(P)\right).
$$
By (\ref{thme2}), we get that $g(\varphi^{j}(Q_{i,1}))=0$ for $0\leq j\leq n-u-1$ and $1\leq i\leq N$. So we obtain
$$
g\in\mathcal{L}\left(\varphi^{-u}(P)+r\sum_{i=0}^{u-1}\varphi^{-i}(P)-\sum_{i=1}^N\sum_{j=0}^{n-u-1}\varphi^{j}(Q_{i,1})\right).
$$
It follows from the fact $g\neq 0$ that
$$
{\rm deg}\left(\varphi^{-u}(P)+r\sum_{i=0}^{u-1}\varphi^{-i}(P)-\sum_{i=1}^N\sum_{j=0}^{n-u-1}\varphi^{j}(Q_{i,1})\right)\geq 0.
$$
Therefore, $1+ur\geq (n-u)N$, i.e., $u\geq \frac{nN-1}{N+r}$.

$Case\ 2:$ If $d\leq n\leq dM$, then we get that
$$
z(\varphi^{j+u}(Q_{i,l}))-f(z(\varphi^{j}(Q_{i,l})),z(\varphi^{j+1}(Q_{i,l})),\cdots,z(\varphi^{j+u-1}(Q_{i,l})))=0,
$$
for $0\leq j\leq d-u-1$, $1\leq l\leq\lfloor\frac{n}{d}\rfloor$ and $1\leq i\leq N$, where the equality follows from only part of the cases of (\ref{thme1}). Using the same argument as in the computation of Case 1, we derive that there exists a nonzero function $g$ such that
$$
g\in\mathcal{L}\left(\varphi^{-u}(P)+r\sum_{i=0}^{u-1}\varphi^{-i}(P)-\sum_{i=1}^N\sum_{j=0}^{d-u-1}\sum_{l=1}^{\lfloor\frac{n}{d}\rfloor}\varphi^{j}(Q_{i,l})\right).
$$
Then, we have $u\geq\frac{dN\lfloor\frac{n}{d}\rfloor-1}{N\lfloor\frac{n}{d}\rfloor+r}$, which completes the proof of the theorem.
\end{proof}

\begin{remark}\label{r1}
In Theorem \ref{thm1}, the lower bound on $N_r(\mathcal{S}_n)$ is of order of magnitude $nN/(N+r)$ if $1<n\leq d$ and the lower bound on $N_r(\mathcal{S}_n)$ is of order of magnitude $nN/(N\lfloor\frac{n}{d}\rfloor+r)$ if\ $d\leq n\leq dM$. In the latter case, the order of magnitude of the lower bound is large when we take a large $d$. For instance, if $q\equiv1\mod 5$, we set $d=(q-1)/5$. Then the lower bound is of order of magnitude $q/r$.
\end{remark}

\section{Multisequences from Hermitian function fields}\label{S4}

In this section, we present two constructions of multisequences arising from the Hermitian function fields and evaluate the lower bound on the joint nonlinear complexity of the multisequences.

\subsection{The first construction of multisequences}

Recall that $H=\Ff_{q^2}(x,y)$ is the Hermitian function field and it genus $g=\frac{q^2-q}{2}$. There is an cyclic group $R$ of order $q^2-1$ generated by an automorphism $\sigma=\varphi_{\delta,0,0}$ of $H$. From Lemma \ref{lem3}, under the action of $R$ on all rational places of $H$, there exists $q$ orbits each containing $q^2-1$ distinct rational places. Let $M,N$ be positive integers with $MN=q-1$. We write all the $q$ orbits by
\begin{eqnarray*}
&&P,\sigma(P),\cdots,\sigma^{q^2-2}(P);\\
&&Q_{1,1},\sigma(Q_{1,1}),\cdots,\sigma^{q^2-2}(Q_{1,1});\\
&&Q_{1,2},\sigma(Q_{1,2}),\cdots,\sigma^{q^2-2}(Q_{1,2});\\
&&\ \ \ \ \ \ \ \ \ \ \ \ \ \ \ \ \ \vdots\\
&&Q_{1,M},\sigma(Q_{1,M}),\cdots,\sigma^{q^2-2}(Q_{1,M});\\
&&\ \ \ \ \ \ \ \ \ \ \ \ \ \ \ \ \ \vdots\\
&&Q_{N,M},\sigma(Q_{N,M}),\cdots,\sigma^{q^2-2}(Q_{N,M}).\\
\end{eqnarray*}
Assume that $P_{\infty}$ is the unique pole of $x$. It follows from the Riemann-Roch Theorem that $\mathcal{L}((2g-1)P_\infty+P)$ is a $(g+1)$-dimensional vector space over $\Ff_{q^2}$. Then, we can choose a function $z\in\mathcal{L}((2g-1)P_\infty+P)$ such that $(z)_\infty=kP_\infty+P$ with some $k\leq 2g-1$. Below, we state our construction of multisequences.

We define a set as
\begin{equation}\label{con2}
\mathcal{S}=\{\mathbf{s}_i:i=1,2,\cdots,N\},
\end{equation}
where
\begin{eqnarray*}
\mathbf{s}_i&=&\{s_i(j)\}_{j=0}^{M(q^2-1)-1}\\
&=&(z(Q_{i,1}),z(\sigma(Q_{i,1}))\cdots,z(\sigma^{q^2-2} (Q_{i,1})),\cdots,z(Q_{i,M}),\cdots,z(\sigma^{q^2-2} (Q_{i,M}))).
\end{eqnarray*}
Thus, $\mathcal{S}$ is a multisequence of length $M(q^2-1)$ over $\Ff_{q^2}$ and its dimension is $N$. In addition, $\mathbf{s}_i$ is a periodic sequence with the period $q^2-1$ if $M=1$. Next, we determine the lower bound of the joint nonlinear complexity for $\mathcal{S}$.

\begin{thm}\label{thm2}
Let $q\geq 5$ be a prime power and $N,M$ positive integers with $NM=q-1$. Let $\mathcal{S}$ be the multisequence defined by (\ref{con2}). Then,
$$
N_r(\mathcal{S}_n)\geq\left\{
            \begin{array}{ll}
              \frac{nN-1}{N+(q^2-q)r},& \hbox{if\ $1<n\leq q^2-1$,}\\
              \frac{(q^2-1)N\left\lfloor n/(q^2-1)\right\rfloor-1}{N\left\lfloor n/(q^2-1)\right\rfloor+(q^2-q)r},& \hbox{if\ $q^2-1\leq n\leq M(q^2-1)$,}
            \end{array}
          \right.
$$
for any integer $r$ with $1\leq r\leq q^2-1$.
\end{thm}
\begin{proof}
Firstly, we prove that $\mathbf{s}_i$ is not a zero sequence for $1\leq i\leq N$. Notice that $(z)_\infty=kP_\infty+P$ with some $k\leq 2g-1$ which implies that
$$
{\rm deg}((z)_\infty)=k+1\leq 2g=q^2-q.
$$
For any $1\leq i\leq N$, if $\mathbf{s}_i$ is a zero sequence, then each $\sigma^j(Q_{i,l})$ is a zero of the function $z$ with $0\leq j\leq q^2-2$, $1\leq i\leq N$ and $1\leq l\leq M$. So we have
$$
{\rm deg}((z)_0)\geq (q^2-1)MN=(q^2-1)(q-1).
$$
Based on the fact that ${\rm deg}((z)_\infty)={\rm deg}((z)_0)$, we deduce that
$$
(q^2-1)(q-1)\leq{\rm deg}((z)_0)={\rm deg}((z)_\infty)\leq q^2-q.
$$
This leads to a contradiction. Consequently, $\mathbf{s}_i$ is not a zero sequence for $1\leq i\leq N$.

Suppose that $n>1$ is an integer and $u$ is an integer with $1\leq u\leq n-1$. Let $f\in\Ff_{q^2}[x_1,\cdots,x_u]$ be a polynomial whose degree is at most $r$ in each variable such that
\begin{equation}\label{thm2e1}
s_i(j+u)=f(s_i(j),s_i(j+1),\cdots,s_i(j+u-1)),
\end{equation}
for $0\leq j\leq n-u-1$ and $1\leq i\leq N$.

$Case\ 1:$ For $1<n\leq q^2-1$, By (\ref{thm2e1}), we derive that
$$
z(\sigma^{j+u}(Q_{i,1}))-f(z(\sigma^{j}(Q_{i,1})),z(\sigma^{j+1}(Q_{i,1})),\cdots,z(\sigma^{j+u-1}(Q_{i,1})))=0,
$$
for $0\leq j\leq n-u-1$ and $1\leq i\leq N$. It follows from Lemma \ref{lem1} that
\begin{equation}\label{thm2e2}
\left(\sigma^{-u}(z)-f(z,\sigma^{-1}(z),\cdots,\sigma^{-u+1}(z))\right)(\sigma^{j}(Q_{i,1}))=0,
\end{equation}
for $0\leq j\leq n-u-1$ and $1\leq i\leq N$.

Set $h=\sigma^{-u}(z)-f(z,\sigma^{-1}(z),\cdots,\sigma^{-u+1}(z))$. Since $(z)_\infty=kP_\infty+P$, we get that
$$
v_{\sigma^{-u}(P)}(\sigma^{-u}(z))=-1, \ \ \ v_{\sigma^{-u}(P)}(f(z,\sigma^{-1}(z),\cdots,\sigma^{-u+1}(z))\geq 0,
$$
which implies that $h\neq 0$. Note that $\sigma$ preserves $P_{\infty}$. It can be easily seen that $v_{P_{\infty}}(\sigma^{-t}(z))=-k\geq-(2g-1)$ for any integer $t\geq 0$. Hence, we have
$$
h\in\mathcal{L}\left((2g-1)ruP_\infty+\sigma^{-u}(P)+r\sum_{i=0}^{u-1}\sigma^{-i}(P)\right).
$$
According to (\ref{thm2e2}), we deduce that $\sigma^{j}(Q_{i,1})$ is a zero of $h$ for $0\leq j\leq n-u-1$ and $1\leq i\leq N$, which implies that
$$
h\in\mathcal{L}\left((2g-1)ruP_\infty+\sigma^{-u}(P)+r\sum_{i=0}^{u-1}\sigma^{-i}(P)-\sum_{i=1}^N\sum_{j=0}^{n-u-1}\sigma^{j}(Q_{i,1})\right).
$$
Due to $h\neq 0$, we obtain
$$
{\rm deg}\left((2g-1)ruP_\infty+\sigma^{-u}(P)+r\sum_{i=0}^{u-1}\sigma^{-i}(P)-\sum_{i=1}^N\sum_{j=0}^{n-u-1}\sigma^{j}(Q_{i,1})\right)\geq 0,
$$
namely,
$$
(2g-1)ru+1+ru\geq N(n-u).
$$
Therefore, $u\geq\frac{nN-1}{N+(q^2-q)r}$.

$Case\ 2:$ For $q^2-1\leq n\leq M(q^2-1)$, by applying only part of the cases of (\ref{thm2e1}), proceeding as in the proof of Case 1, we get that $h=\sigma^{-u}(z)-f(z,\sigma^{-1}(z),\cdots,\sigma^{-u+1}(z))$ is a nonzero function that belongs to
$$
\mathcal{L}\left((2g-1)ruP_\infty+\sigma^{-u}(P)+r\sum_{i=0}^{u-1}\sigma^{-i}(P)-\sum_{i=1}^N\sum_{j=0}^{q^2-u-2}\sum_{l=1}^{\lfloor n/(q^2-1)\rfloor}\sigma^{j}(Q_{i,l})\right).
$$
So the degree of the divisor
$$
(2g-1)ruP_\infty+\sigma^{-u}(P)+r\sum_{i=0}^{u-1}\sigma^{-i}(P)-\sum_{i=1}^N\sum_{j=0}^{q^2-u-2}\sum_{l=1}^{\lfloor n/(q^2-1)\rfloor}\sigma^{j}(Q_{i,l})
$$
is nonnegative, i.e., $2gru+1\geq N(q^2-u-1)\lfloor n/(q^2-1)\rfloor$. Hence, $u\geq\frac{(q^2-1)N\left\lfloor n/(q^2-1)\right\rfloor-1}{N\left\lfloor n/(q^2-1)\right\rfloor+(q^2-q)r}$. This completes the proof of the theorem.
\end{proof}
\begin{remark}\label{r2}
It can be seen that the lower bound on $N_r(\mathcal{S}_n)$ in Theorem \ref{thm2} is of order of magnitude $nN/(rq^2)$. By maximizing the parameter $n$, the order of magnitude of the lower bound is $q/r$.
\end{remark}

\subsection{The second construction of multisequences}

By Lemma \ref{lem4}, the automorphism group $G$ generated by $\phi$ divides all rational places $\neq P_\infty$ of $H$ into $q^3/p$ orbits and every orbit contains $p$ distinct elements. Label all the $q^3/p$ orbits as follows:
\begin{eqnarray*}
&&P,\phi(P),\cdots,\phi^{p-1}(P);\\
&&Q_{1,1},\phi(Q_{1,1}),\cdots,\phi^{p-1}(Q_{1,1});\\
&&Q_{1,2},\phi(Q_{1,2}),\cdots,\phi^{p-1}(Q_{1,2});\\
&&\ \ \ \ \ \ \ \ \ \ \ \ \ \ \ \ \ \vdots\\
&&Q_{1,M},\phi(Q_{1,M}),\cdots,\phi^{p-1}(Q_{1,M});\\
&&\ \ \ \ \ \ \ \ \ \ \ \ \ \ \ \ \ \vdots\\
&&Q_{N,M},\phi(Q_{N,M}),\cdots,\phi^{p-1}(Q_{N,M}),\\
\end{eqnarray*}
where $N,M$ are positive integers with $NM=q^3/p-1$. Let $P_{\infty}$ be the unique pole of $x$. Clearly, $\phi(P_{\infty})=P_{\infty}$. According to the Riemann-Roch Theorem, we obtain that $\mathcal{L}((2g-1)P_\infty+P)$ is a vector space over $\Ff_{q^2}$ and its dimension is $g+1$. So there is a function $z\in\mathcal{L}((2g-1)P_\infty+P)$ such that $(z)_\infty=kP_\infty+P$ with some $k\leq 2g-1$. Define a set of sequences by
\begin{equation}\label{con3}
\mathcal{S}=\{\mathbf{s}_i:i=1,2,\cdots,N\},
\end{equation}
where
\begin{eqnarray*}
\mathbf{s}_i&=&\{s_i(j)\}_{j=0}^{Mp-1}\\
&=&(z(Q_{i,1}),z(\phi(Q_{i,1}))\cdots,z(\phi^{p-1} (Q_{i,1})),\cdots,z(Q_{i,M}),\cdots,z(\phi^{p-1} (Q_{i,M}))).
\end{eqnarray*}
Then $\mathcal{S}$ is a multisequence of dimension $N$ over $\Ff_{q^2}$. The length of each sequence $\mathbf{s}_i$ of $\mathcal{S}$ is $Mp$. What is more, for any $1\leq i\leq N$, $\mathbf{s}_i$ is a periodic sequence with the period $p$ if $M=1$.

\begin{thm}\label{thm3}
Suppose that $p$ is an odd prime and $q$ is a power of $p$. Let $N,M$ positive integers with $NM=q^3/p-1$. Let $\mathcal{S}$ be the multisequence defined by (\ref{con3}). Then, we have
$$
N_r(\mathcal{S}_n)\geq\left\{
            \begin{array}{ll}
              \frac{nN-1}{N+(q^2-q)r},& \hbox{if\ $1<n\leq p$,}\\
              \frac{pN\left\lfloor n/p\right\rfloor-1}{N\left\lfloor n/p\right\rfloor+(q^2-q)r},& \hbox{if\ $p\leq n\leq Mp$,}
            \end{array}
          \right.
$$
for any integer $r$ with $1\leq r\leq q^2-1$.
\end{thm}
\begin{proof}
Using the same argument as in the proof of Theorem \ref{thm2}, the desired result follows. We omit the details.
\end{proof}
\begin{remark}\label{r3}
For the case that $1<n\leq p$, if the order of magnitude of $N$ is higher than that of $q^2r$, then the lower bound on $N_r(\mathcal{S}_n)$ is of order of magnitude $n$. Otherwise, the lower bound is of order of magnitude $nN/q^2r$. For the case that $p\leq n\leq Mp$, the lower bound on $N_r(\mathcal{S}_n)$ is of order of magnitude $p$ if the order of magnitude of $N\left\lfloor n/p\right\rfloor$ is higher than that of $q^2r$. On the other hand, the lower bound is of order of magnitude $nN/q^2r$ if the order of magnitude of $q^2r$ is higher than that of $N\left\lfloor n/p\right\rfloor$. If we take the maximum $n=Mp$, then the order of magnitude of the lower bound is $q/r$.
\end{remark}

\section{Concluding remarks}\label{S5}

Let $\mathcal{R}$ be a random multisequence over $\Ff_q$ of length $n$ and dimension $m$. Meidl and Niederreiter \cite{Nonlinear} has pointed out that the expected order of magnitude of $N_r(\mathcal{R}_n)$ is ${\rm log}(mn)$ under the heuristic method. According to Remark \ref{r1}, Remark \ref{r2} and Remark \ref{r3}, it is obvious that the multisequences constructed in Theorem \ref{thm1}, Theorem \ref{thm2} and Theorem \ref{thm3} can be said to have high joint nonlinear complexity in appropriate cases on the parameters of the multisequences.

\end{document}